\documentclass[12pt]{article}
\usepackage{amssymb}
\usepackage{amsmath}
\PassOptionsToPackage{hyphens}{url}
\usepackage[colorlinks,linkcolor=blue,citecolor=green]{hyperref}

\newtheorem{theorem}{Theorem}

\newtheorem{corollary}[theorem]{Corollary}

\newtheorem{lemma}[theorem]{Lemma}
\newtheorem{proposition}[theorem]{Proposition}
\newtheorem{remark}[theorem]{Remark}
\newenvironment{proof}[1][Proof]{\textbf{#1.} }{\ \rule{0.5em}{0.5em}}
\newcommand*{\pd}
[2]{\mathchoice{\frac{\partial#1}{\partial#2}}
  {\partial#1/\partial#2}{\partial#1/\partial#2}
  {\partial#1/\partial#2}}

\newcommand*{\fd}
[2]{\mathchoice{\frac{\delta#1}{\delta#2}}
  {\delta #1/\delta#2}{\delta#1/\delta#2}{\delta#1/\delta#2}}

\begin{document}
\title{Projective geometry of homogeneous\\
  second order Hamiltonian operators}
\author{Pierandrea Vergallo$^{1,3}$, \qquad Raffaele Vitolo$^{2,3}$ \\[3mm]
  \small $^1$  Department of Mathematical, Computer,\\
  \small Physical and Earth Sciences\\
  \small University of Messina (Italy), \\
  \small V.le F. Stagno D’Alcontres 31, I-98166 Messina, Italy\\
 \small $^2$ Department of Mathematics and Physics \textquotedblleft E. De
  Giorgi\textquotedblright ,\\
  \small Universit\`a del Salento, Lecce, Italy\\
 \small $^3$ Istituto Nazionale di Fisica Nucleare, Sez.\ Lecce
  \\
  \texttt{pierandrea.vergallo@unime.it}
  \\
  \texttt{raffaele.vitolo@unisalento.it} }
\maketitle
\begin{abstract}
  We prove the invariance of homogeneous second-order Hamiltonian operators
  under the action of projective reciprocal transformations. We establish a
  correspondence between such operators in dimension $n$ and $3$-forms in
  dimension $n+1$. In this way we classify second order Hamiltonian operators
  using the known classification of $3$-forms in dimensions $\leq 9$. Systems
  of first-order conservation laws that are Hamiltonian with respect to such
  operators are also explicitly found. The integrability of the systems is
  discussed in detail.
\end{abstract}

\newpage

\tableofcontents

\allowdisplaybreaks[3]

\section{Introduction}

Hamiltonian operators play a fundamental role in the theory of integrability of
partial differential equations (PDEs). A distinguished class of Hamiltonian
operators was introduced in 1983 by Dubrovin and Novikov~\cite{DN83}. One of
the main features of the new class was that it was invariant under
diffeomorphisms of the underlying space, thus bringing geometry into the theory
of integrable PDEs.

More precisely, let us denote by $u^i=u^i(t,x)$ $n$ unknown functions of two
independent variables $t$ and $x$, $i=1$, \dots, $n$, and denote by
$u^i_\sigma$ the $x$-derivative of $u^i$ $\sigma$ times.  An element of the
class has the form of a matrix of first-order differential operators where each
summand contains the same number of $x$-derivatives ($1$ in this case):
\begin{equation}
  \label{eq:27}
  P^{ij}=g^{ij}\partial_x + \Gamma^{ij}_k u^k_x.
\end{equation}
Then, a partial differential equation of the form
$u^i_t = f^i(u^j,u^j_\sigma)$ is Hamiltonian if there exists a density
$H=\int h(u^j)dx$ such that
\begin{equation}
  \label{eq:31}
  u^i_t = f^i = P^{ij}\fd{H}{u^j}.
\end{equation}
The typical situation is when $f^i = V^i_j(u^k)u^j_x$, \emph{i.e.} the system
of PDEs is quasilinear and of the first order. In that case, if the system of
PDEs is Hamiltonian, it is straightforward to realize that both the system and
its Hamiltonian formulation are form-invariant with respect to a transformation
involving the dependent variables only: $\bar{u}^i=\bar{u}^i(u^j)$. We recall
that the Hamiltonian property of $P$ is the fact that it induces a Poisson
bracket on the space of densities:
\begin{equation}
  \label{eq:33}
  \{F,G\}_P = \int \fd{f}{u^i}P^{ij}\fd{g}{u^j}dx,
\end{equation}
where $F=\int f\, dx$ and $G=\int g\,dx$. The Poisson bracket property is also
invariant with respect to the above transformations.

The first order case was soon generalized to the higher order case in
\cite{DubrovinNovikov:PBHT}, the homogeneity degree being equal to the order of
the operators. Numerous examples shown that homogeneous Hamiltonian operators
(HHOs) are ubiquitous, either as stand-alone operators or in linear
combinations with operators of different homogeneity degrees. See
\cite{LSV:bi_hamil_kdv} for many examples of the latter kind.

Recently, it was observed that third-order homogeneous Hamiltonian operators
are invariant (when in their canonical form, see \cite[eq.\ (2)]{FPV14}) with
respect to a non-obvious class of transformations, namely projective reciprocal
transformations \cite{FPV14}. They have the form of a projective transformation
of dependent variables coupled with a non-local transformation of $x$:
\begin{equation}
  \label{eq:32}
  \tilde{u}^i = \frac{A^i_j u^j + A^0_j}{A^0_j u^j + A^0_0},\qquad
  d\tilde{x} = (A^0_j u^j + A^0_0)dx.
\end{equation}
Third-order homogeneous Hamiltonian operators have been classified in low
dimensions ($n\leq 4$, \cite{FPV14,FPV16}) with respect to the action of the
above group. It was also proved that there exists a multiparametric family of
quasilinear systems of first-order PDEs that are Hamiltonian with respect to
any such operator \cite{FPV17:_system_cl}.

\emph{The goal of the current paper is to prove that second-order homogeneous
Hamiltonian operators have the same invariance properties of third-order
homogeneous Hamiltonian operators.}

The interest in such a result is that projective-geometric invariance is not
just an `isolated' feature of third-order operators: being also a property of
second-order operators it is reasonable to think that it is something bound to
all homogeneous operators.

A classification then follows from the above invariance result; the algebraic
variety that is identified with a second-order operator has been extensively
studied in \cite{DePoiFaenziMezzettiRanestad17} (linear line congruence), and
has a different nature with respect to that which is associated with a
third-order operator (quadratic line complex). We also obtain similar (but not
identical!) results concerning associated systems of first-order PDEs.

More precisely, second order homogeneous Hamiltonian operators have the general
form
\begin{equation}
  P^{ij}=g^{ij}\partial_x^2+b^{ij}_ku^k_x\partial_x+c^{ij}_ku^{k}_{xx}
  +c^{ij}_{kh}u^k_xu^h_x.
\end{equation}
We will always consider the non-degenerate case $\det (g^{ij})\neq 0$.  Under a
coordinate transformation of the type $\bar{u}^i = \bar{u}^i(u^j)$ the symbols
$\Gamma_{ij}^k = - g_{ip}c^{pk}_j$ transform as a linear connection.  It is
proved in~\cite{potemin86:_poiss,doyle93:_differ_poiss} (but see
also~\cite{mokhov98:_sympl_poiss,ferguson07:_flat_hamil}) that the Hamiltonian
property of the above operator implies that $\Gamma_{ij}^k$ is symmetric and
flat. With respect to flat coordinates the operator can be rewritten as
\begin{equation}\label{1b}
  P^{ij}=\partial_xg^{ij}\partial_x.
\end{equation}
The Hamiltonian property in flat coordinates is then equivalent to the fact
that
\begin{equation}
  \label{aa3}g_{ij}=T_{ijk}u^k+g_{ij}^0,
\end{equation}
where $T_{ijk}$ are constant and skew-symmetric with respect to $i$, $j$, $k$
and $g^0_{ij}$ is constant and skew-symmetric with respect to $i$, $j$.
The above equations~\eqref{1b}, \eqref{aa3} have been independently found in
\cite{potemin86:_poiss,doyle93:_differ_poiss}. See
also~\cite{mokhov98:_sympl_poiss} for a thorough review on homogeneous
Hamiltonian operators, and see \cite{ferguson07:_flat_hamil} for a further
differential-geometric analysis of the properties of homogeneous second-order
Hamiltonian operators and their pencils.

Here, we will prove the following theorem (Theorem~\ref{t2} on page
\pageref{t2}).
\begin{theorem}
  Second-order homogeneous Hamiltonian operators are invariant under projective
  reciprocal transformations~\eqref{eq:32}.
\end{theorem}

Then, we will prove a result that enables us to classify second-order
homogeneous Hamiltonian operators in low dimensions ($n\leq 8$)
(Theorem~\ref{th:corresp} on page \pageref{th:corresp}).

\begin{theorem}
  Second-order homogeneous Hamiltonian operators in dimension $n$ can be put in
  bijection with $3$-forms in the $n+1$-dimensional space $\mathbb{C}^{n+1}$.

  A projective reciprocal transformation induces an $SL(n+1)$-transformation on
  the corresponding $3$-form which commutes with the action on the
  corresponding second-order homogeneous Hamiltonian operator.
\end{theorem}
It is interesting to observe that, in a generic situation, $3$-forms define
\emph{linear line congruences} (see \cite{DePoiFaenziMezzettiRanestad17} for
the algebraic geometric description and properties of that correspondence),
hence second-order homogeneous Hamiltonian operators are in correspondence with
algebraic varieties, as it happened in the third-order case (for different
algebraic varieties, \emph{i.e.} quadratic line complexes).

Let us consider a quasilinear first-order system of PDEs in conservative form
\begin{equation}\label{2a}
  u^i_t = (V^i)_x,
\end{equation}
where $V^i = V^i(u^j)$. We will prove the following Theorem (which is the union
of the statements of Theorem~\ref{th:solut-comp-cond},
Corollary~\ref{cor:lindeg}, Proposition~\ref{pr:hamiltonian-systems-2} and
Theorem~\ref{haant}).
\begin{theorem}
  Every second-order homogeneous Hamiltonian operator is the Hamiltonian
  operator of a multiparameter family of systems of conservation laws as
  in~\eqref{2a}. The systems are linearly degenerate, possess a non-local
  Hamiltonian and have identically vanishing Haantjes tensor.
\end{theorem}
Strictly speaking, we do not have a general proof of integrability of the
systems, but it turns out that in many examples they are indeed integrable.
Indeed, in many computational experiments where we randomly chose the
parameters it turned out that the systems are diagonalizable, and, by known
results \cite{sevennec93:_geomet}, they have a maximal set of conservation
laws. Moreover, until now we did not find any counterexamples to integrability
in our experiments. See the discussion at the end of
Subsection~\ref{sec:haantj-tens-integr}.

\medskip

We would like to stress that homogeneous Hamiltonian operators are important
building blocks in the theory of integrable systems. We can mention several
ways in which they are involved:
\begin{itemize}
\item Many bi-Hamiltonian systems have a bi-Hamiltonian pair of the form
  $P=P_1 + R$ and $Q=Q_1$, where $P_1$, $Q_1$ are compatible homogeneous
  first-order Hamiltonian operators and $R$ is a homogeneous second-order or
  third-order Hamiltonian operator which is compatible with $P_1$ and $Q_1$. We
  call such systems \emph{bi-Hamiltonian systems of KdV
    type}~\cite{LSV:bi_hamil_kdv}. Examples include the AKNS (or two-boson)
  hierarchy, the two-component Camassa-Holm
  hierarchy~\cite{falqui06:_camas_holm}, a multiparameter family defined
  in~\cite{sole15:_hamil_pdes} (see \cite{LSV:bi_hamil_kdv}) and the
  Kaup--Broer system~\cite{Kuper_85} when
  \begin{equation}
    \label{eq:34}
    R=
    \begin{pmatrix}
      0 & -1 \\ 1 & 0
    \end{pmatrix}\partial_x^2.
  \end{equation}
  Other examples with $R$ a third-order HHO are the KdV equation, the
  Camassa-Holm equation~\cite{LSV:bi_hamil_kdv}, a dispersive water waves
  equation~\cite{antonowicz89:_factor_scroed} and a coupled Harry--Dym
  hierarchy~\cite{antonowicz88:_coupl_harry_dym}. See also
  the recent papers \cite{bolsinov:_applic_nijen_iii,konyaev:_geomet_poiss},
  with a differential-geometric focus on the same construction.
\item Fewer systems are determined by a pair of Hamiltonian operators of the
  form $P=P_1$ and $R$; here we mention the WDVV systems
  \cite{vasicek21:_wdvv_hamil}, where $R$ is of the third order. No instances
  of systems that we determined in this paper were previously known to our
  knowledge. A probable explanation is that the first non-trivial systems
  (although linearizable) appear in dimension $4$, and non-linearizable ones in
  dimension $6$ and greater, and that makes their investigation quite
  complicated.
\item Homogeneous operators play a central role in Dubrovin--Zhang's
  perturbative approach to the classification of integrable systems under the
  action of the group of Miura transformations
  \cite{dubrovin01:_normal_pdes_froben_gromov_witten}. Since deformations of a
  first-order Poisson pencil are given as a formal series of homogeneous
  operators, one might expect that projective transformations and invariance
  can play a role.
\end{itemize}

The results obtained so far show that \emph{the group of projective reciprocal
  transformations act on hierarchies defined by trios of compatible operators
  $P_1$, $Q_1$, $R$ or by pairs of compatible operators $P_1$, $R$}. The action
preserves the locality of second-order or third-order HHOs in canonical form,
even if it does not preserve the locality of $P_1$, $Q_1$. So, a projective
geometric study of the above hierarchies makes sense and is potentially
interesting. 

The chances that the projective invariance properties that are shared by
second-order and third-order HHOs might be generalized to HHOs of arbitrary
order are high enough to consider that possibility in the framework of the
perturbative approach in~\cite{dubrovin01:_normal_pdes_froben_gromov_witten}.

More generally, our results might indicate that a projective-geometric approach
to integrable systems is starting to emerge in the field. We will pursue that
research line in the future.

\section{Projective geometry and Hamiltonian operators}
\label{sec:second-order-homog}

Let us consider a projective transformation
$T\colon\mathbb{P}^n\rightarrow \mathbb{P}^n$. We will treat $(u^i)$ as an
affine chart of the homogeneous coordinates $[\mathbf{v}]=[v^0,\ldots,v^n]$,
where $v^i=u^i/u^{n+1}$. In homogeneous coordinates we have
$T(v)=[a^\lambda_\mu v^\mu]$, where $(a^\lambda_\mu)\in SL(n+1)$. In this
section latin indices $i$, $j$, \dots will run from $1$ to $n$ and greek
indices $\lambda$, $\mu$, \dots will run from $1$ to $n+1$. A projective
transformation in the affine chart has the form:
\begin{equation}\label{22}
  \tilde{u}^i = T^i(u^j) = \frac{a^i_ju^j+a^i_{n+1}}{a^{n+1}_ju^j+a^{n+1}_{n+1}}.
\end{equation}

In this section we will calculate the action of a projective transformation on
a second-order homogeneous Hamiltonian operator. We will realize that the
transformation \eqref{22} alone is not enough to yield invariance, while
reciprocal projective transformations guarantee the invariance of the
form~\eqref{1b} of our operators.

\subsection{Projective invariance of the Hamiltonian operators}
\label{sec:proj-invar-hamilt}

We would like to find the change of coordinates formula on the leading
coefficient $g$ of a second-order homogeneous Hamiltonian operator; in other
words, we are looking for a formula connecting\footnote{In skew-symmetric
  tensors indices are summed on all their ranges, according to Einstein's
  convention.}
\begin{equation}
  \label{eq:4}
  g = (T_{ijk}u^k + g^0_{ij})du^i\wedge du^j \quad\text{and}\quad
  \tilde{g} = (\tilde{T}_{ijk}\tilde{u}^k + \tilde{g}^0_{ij})
  d\tilde{u}^i\wedge d\tilde{u}^j.
\end{equation}
Note that we will work with the covariant version of the leading coefficient;
this is possible due to our assumption $\det(g)\neq 0$.  As a preliminary
remark, note that
\begin{equation}\label{2}
  d\left(\tilde{u}^i\right)=d\left(\frac{a^i_su^s+a^i_{n+1}}
    {a^{n+1}_su^s+a^{n+1}_{n+1}}\right)
  = \frac{Aa^i_sdu^s-(a^i_su^s+a^i_{n+1})a^{n+1}_ldu^l}{A^2}
\end{equation}where $A=a^{n+1}_su^s+a^{n+1}_{n+1}$.

\begin{theorem}
  Under the transformation \eqref{22} we obtain
  \begin{subequations}\label{eq:14}
\begin{align}\label{eq:9}
  \begin{split}
  T_{lcs}=&\frac{1}{2A^3}\Big(\tilde{T}_{ijk}(a^i_la^j_c-a^i_ca^j_l)a^k_s
    +\tilde{g}_{ij}^0(a^i_la^j_c - a^i_ca^j_l)a^{n+1}_s
    \\
    &-\tilde{g}^0_{ij}(a^i_la^{n+1}_c-a^i_ca^{n+1}_l)a^j_s
    -\tilde{g}^0_{ij}(a^{n+1}_la^j_c-a^{n+1}_ca^j_l)a^i_s\Big)
    \end{split}
  \\ \label{eq:13}
  \begin{split}
  g^0_{lc}=&\frac{1}{2A^3}\Big(\tilde{T}_{ijk}(a^i_la^j_c-a^i_ca^j_l)a^k_{n+1}
    +\tilde{g}^0_{ij}(a^i_la^j_c-a^i_ca^j_l)a^{n+1}_{n+1}
    \\
    &-\tilde{g}_{ij}^0(a^i_la^{n+1}_c-a^i_ca^{n+1}_l)a^j_{n+1}
    -\tilde{g}^0_{ij}(a^j_ca^{n+1}_l-a^j_la^{n+1}_c)a^i_{n+1}\Big)
    \end{split}
\end{align}
\end{subequations}
\end{theorem}
\begin{proof}
  Applying the transformation to
  $\tilde{g}_{ij}d\tilde{u}^i\wedge d\tilde{u}^j=
  \tilde{T}_{ijk}\tilde{u}^kd\tilde{u}^i\wedge d\tilde{u}^j+
  \tilde{g}_{ij}^0d\tilde{u}^i\wedge d\tilde{u}^j$ we obtain
\begin{align*}
  \tilde{g}_{ij}^0d\tilde{u}^i&\wedge d\tilde{u}^j=
  \\
  &=\tilde{g}_{ij}^0\left(\frac{Aa^i_sdu^s-(a^i_su^s+a^i_{n+1})a^{n+1}_ldu^l}{A^2}
    \right)\wedge \left(\frac{Aa^j_sdu^s-(a^j_su^s+a^j_{n+1})a^{n+1}_ldu^l}{A^2}\right)
  \\
  &=\frac{\tilde{g}^0_{ij}}{A^4}\big(A^2a^i_sa^j_ldu^s\wedge
    du^l-Aa^i_s(a^j_bu^b+a^j_{n+1})a^{n+1}_cdu^s\wedge du^c
  \\
  &\hphantom{ciao} - A(a^i_mu^m+a^i_{n+1})a^{n+1}_la^j_sdu^l\wedge du^s+
  \\
  &\hphantom{ciao}(a^i_mu^m+a^i_{n+1})a^{n+1}_l(a^j_bu^b+a^j_{n+1})a^{n+1}_cdu^l\wedge
    du^c\big)
  \\
  &=\frac{\tilde{g}^0_{ij}}{A^4}\big(A^2a^i_sa^j_ldu^s\wedge
    du^l-Aa^i_s(a^j_bu^b+a^j_{n+1})a^{n+1}_cdu^s\wedge du^c
  \\
  &\hphantom{ciao}-A(a^i_mu^m+a^i_{n+1})a^{n+1}_la^j_sdu^l\wedge du^s\big)
  \\
  &=\frac{\tilde{g}^0_{ij}}{A^3}\big(Aa^i_sa^j_ldu^s\wedge
       du^l-a^i_s(a^j_bu^b+a^j_{n+1})a^{n+1}_cdu^s\wedge du^c
  \\
  &\hphantom{ciao}-(a^i_mu^m+a^i_{n+1})a^{n+1}_la^j_sdu^l\wedge du^s\big)
\end{align*}
Analogously,
\begin{align*}
  \tilde{T}_{ijk}\tilde{u}^kd\tilde{u}^i\wedge d\tilde{u}^j&=
  \tilde{T}_{ijk}\frac{a^i_su^s+a^i_{n+1}}{A^5}A^2a^i_la^j_cdu^l\wedge du^c
  \\
  &=\frac{\tilde{T}_{ijk}}{A^3}(a^k_su^s+a^k_{n+1})a^i_la^j_cdu^l\wedge du^c,
\end{align*}
where three terms cancel due to the skew-symmetry of $\tilde{T}_{ijk}$.  We
finally obtain
\begin{multline}
  \tilde{g}_{ij}d\tilde{u}^i\wedge d\tilde{u}^j=
  \frac{1}{A^3}\Big[\tilde{T}_{ijk}(a^k_su^s+a^k_{n+1})a^i_la^j_c +
\tilde{g}^0_{ij}a^i_l(a^{n+1}_su^s+a^{n+1}_{n+1})a^j_c \\
 -\tilde{g}^0_{ij}a^i_l(a^j_su^s+a^j_{n+1})a^{n+1}_c
 -\tilde{g}^0_{ij}a^j_c(a^i_su^s+a^i_{n+1})a^{n+1}_l\Big]du^l\wedge du^c
\end{multline}
Collecting $u^s$, and comparing the left-hand side with $g_{lc}du^l\wedge du^c$
with respect to a basis (\emph{i.e.}, $l<c$) we obtain the change of
coordinates formula~\eqref{eq:14}
\end{proof}

\begin{corollary}\label{cor:proj-invar-hamilt-1}
  The indexed families $T_{lcs}$ and $g^0_{lc}$ as obtained from
  $\tilde{T}_{ijk}$ and $\tilde{g}^0_{ij}$ by means of the above transformation
  are skew-symmetric with respect to all of their indices. Hence, a projective
  transformation of the leading coefficient of a second-order HHO preserves its
  form up to a conformal factor:
  \begin{equation}
    \label{eq:15}
    \tilde{g}_{ij}d\tilde{u}^i\wedge d\tilde{u}^j=\frac{1}{A^3}
    g_{lc}du^l\wedge du^c.
  \end{equation}
\end{corollary}
\begin{proof}
  The skew-symmetry of $g^0_{lc}$ is evident, and it is easy to show that
  $T_{lcs}=-T_{lsc}$ by observing that the skew-symmetry holds separately in
  the summand $\tilde{T}_{ijk}(a^i_la^j_c-a^i_ca^j_l)a^k_s$ and in the
  remaining three summands.
\end{proof}

We recall that a reciprocal transformation is a nonlocal change of the
independent variables $t$, $x$ defined as
\begin{equation}\label{rec1}
  d\tilde{t} = B(u)dt,\qquad d\tilde{x}=A(u)dx
\end{equation}
where $A(u)$, $B(u)$ are functions depending on $u$. Projective reciprocal
transformation were introduced in \cite{FPV14} as invariance transformations
for the canonical form of third-order HHOs. They are reciprocal transformations
of the form
\begin{equation}\label{eq:7}
  d\tilde{t} = dt,\qquad d\tilde{x} = A\,dx = (a^{n+1}_ku^k+a^{n+1}_{n+1})dx
\end{equation}
coupled with a projective transformation $T$ as in~\eqref{22}. We are going to
prove that projective reciprocal transformations preserve the canonical
form~\eqref{1b} of second-order HHOs. The proof follows the lines of the
proof of the analogous result for third-order HHOs \cite{FPV14}.

\begin{theorem}\label{t2}
  Projective reciprocal transformations preserve the canonical form \eqref{1b}
  of second-order HHOs.
\end{theorem}
\begin{proof}
  It is enough to prove the result for a transformation of the
  type $\tilde{u}^i = u^i/A$, where $A=a^{n+1}_ku^k+a^{n+1}_{n+1}$. It is easy
  to see that $\int{u^idx}$ transform as
  $\int{\frac{u^i}{A}d\tilde{x}}=\int{\tilde{u}^id\tilde{x}}$, and, more
  generally, two densities $F=\int{f(u)dx}$ and $H=\int{h(u)dx}$ transform as
  $f=A\tilde{f}$ and $h=A\tilde{h}$.  Moreover, we have:
  \begin{equation}
    f_j=\frac{\partial f}{\partial u^j}=\frac{\partial A}{\partial u^j}
    \tilde{f}+A\frac{\partial \tilde{f}}{\partial u^j}
    =a^{n+1}_j\tilde{f}+A\tilde{f}_j
  \end{equation}
  and analogously $h_j=a^{n+1}_j\tilde{h}+A\tilde{h}_j$, then:
  \begin{multline}\label{eq:29}
    \{F,H\}=\displaystyle \int{f_iP^{ij}h_jdx}=\\
            =\displaystyle\int{(a^{n+1}_i\tilde{f}+A\tilde{f}_i)A
              \partial_{\tilde{x}}\left(g^{ij}A\partial_{\tilde{x}}
              (a^{n+1}_j\tilde{h}+A\tilde{h}_j)\right)\frac{1}{A}d\tilde{x}},
  \end{multline}
  where we used $\partial_x=A\partial_{\tilde{x}}$.  We can cancel $A$ once and
  obtain a new second order HHO with leading term $A^3g^{ij}$.  Let us first
  observe that
  $
  a^{n+1}_{n+1}+a^{n+1}_ku^k=A=a^{n+1}_{n+1}\frac{1}{1-a^{n+1}_l\tilde{u}^l}$.
  Then, we have
\begin{gather}
  \frac{\partial \tilde{u}^i}{\partial u^j}=
  \frac{\delta^i_jA-a^{n+1}_ju^i}{A^2}=
  \frac{\delta^i_j-a^{n+1}_i\tilde{u}^j}{A}
  \\
  A\frac{\partial \tilde{f}}{\partial u^j}=
  A\frac{\partial \tilde{u}^k}{\partial u^j}
  \frac{\partial \tilde{f}}{\partial \tilde{u}^k}=
  (\delta^k_j-a^{n+1}_j\tilde{u}^k)
  \frac{\partial \tilde{f}}{\partial \tilde{u}^k}
\end{gather}
Now, let us consider again the bracket in~\eqref{eq:29} and carry out the
coordinate change:
\begin{align*}
\{F,H\}=&\int{\left(A\frac{\partial \tilde{f}}{\partial u^i}+a^{n+1}_i\tilde{f}\right)\partial_{\tilde{x}}\left(g^{ij}A\partial_{\tilde{x}}\left(A\frac{\partial \tilde{h}}{\partial u^j}+a^{n+1}_j\tilde{h}\right)\right)d\tilde{x}}\\
=&\int{\left((\delta^k_i-a^{n+1}_i\tilde{u}^k)\frac{\partial \tilde{f}}{\partial \tilde{u}^k}+a^{n+1}_i\tilde{f}\right)\partial_{\tilde{x}}\left(g^{ij}A\partial_{\tilde{x}}\left((\delta^l_j-a^{n+1}_j\tilde{u}^l)\frac{\partial \tilde{h}}{\partial \tilde{u}^l}+a^{n+1}_j\tilde{h}\right)\right)d\tilde{x}}
\end{align*}
Using the identity:
\begin{displaymath}
\partial_{\tilde{x}}\left((\delta^l_j-a^{n+1}_j\tilde{u}^l)\frac{\partial \tilde{h}}{\partial \tilde{u}^l}+a^{n+1}_j\tilde{h}\right)=(\delta^l_j-a^{n+1}_j\tilde{u}^l)\partial_{\tilde{x}}\frac{\partial \tilde{h}}{\partial \tilde{u}^l}
\end{displaymath}
we obtain
\begin{align*}
\{F,H\}=&\int{\left((\delta^k_i-a^{n+1}_i\tilde{u}^k)\frac{\partial
          \tilde{f}}{\partial
          \tilde{u}^k}+a^{n+1}_i\tilde{f}\right)\partial_{\tilde{x}}
          \left(g^{ij}A(\delta^l_j-a^{n+1}_j\tilde{u}^l)\partial_{\tilde{x}}
          \frac{\partial
          \tilde{h}}{\partial \tilde{u}^l}\right)d\tilde{x}}
  \\
  =&\int{\left((\delta^k_i-a^{n+1}_i\tilde{u}^k)
     \frac{\partial \tilde{f}}{\partial \tilde{u}^k}\right)
     \partial_{\tilde{x}}\left(g^{ij}A(\delta^l_j-a^{n+1}_j\tilde{u}^l)
     \partial_{\tilde{x}}\frac{\partial \tilde{h}}{\partial \tilde{u}^l}\right)
     d\tilde{x}}
  \\
        &\hphantom{ciao}-\int{a^{n+1}_i\partial_{\tilde{x}}\tilde{f}\cdot
          \left(g^{ij}A(\delta^l_j-a^{n+1}_j\tilde{u}^l)
          \partial_{\tilde{x}}\frac{\partial \tilde{h}}{\partial \tilde{u}^l}
          \right)d\tilde{x}}.
\end{align*}
Finally, observing that
$\partial_{\tilde{x}}\tilde{f}=\tilde{f}_{,m}\tilde{u}^m_{\tilde{x}}$ and by
using the identity
\begin{displaymath}
(\delta^k_i-a^{n+1}_i\tilde{u}^k)\frac{\partial \tilde{f}}{\partial
  \tilde{u}^k}\partial_{\tilde{x}}-a^{n+1}_i
\frac{\partial \tilde{f}}{\partial \tilde{u}^k}\tilde{u}_{\tilde{x}}^k
=\frac{\partial \tilde{f}}{\partial \tilde{u}^k}
\partial_{\tilde{x}}(\delta^k_i-a^{n+1}_i\tilde{u}^k)
\end{displaymath}
we have
\begin{equation}\label{eq:30}
  \{F,H\}=\int{\frac{\partial \tilde{f}}{\partial \tilde{u}^k}
    \tilde{P}^{kl}\frac{\partial \tilde{h}}{\partial \tilde{u}^l}d\tilde{x}}
\end{equation}
with
\begin{equation}
  \tilde{P}^{kl}=\partial_{\tilde{x}}(\delta^k_i-a^{n+1}_i\tilde{u}^k)
  g^{ij}A(\delta^l_j-a^{n+1}_j\tilde{u}^l)\partial_{\tilde{x}} =
  \partial_{\tilde{x}}\tilde{g}^{ij}\partial_{\tilde{x}}.
\end{equation}
where $\tilde{P}$ is again a local homogeneous operator of second order in view
of Corollary~\ref{cor:proj-invar-hamilt-1}.
\end{proof}

\subsection{Projective interpretation of the Hamiltonian operators}
\label{sec:proj-class-hamilt}

The action of the projective group on second-order HHOs allows us to classify
such operators. Indeed, we exhibit a bijective correspondence of the leading
term of the operator (in dimension $n$) with a projective $3$-form (in
dimension $n+1$). Such geometric objects are well-known in algebraic
geometry~\cite{DePoiFaenziMezzettiRanestad17} and there exist a classification
in dimensions up to $n+1=9$. Of course, we are interested in the even cases
$n=2$, $4$, $6$, $8$ due to the assumption $\det(g)\neq 0$.

Let us set
\begin{equation}\label{eq:8}
T_{n+1\,jk}=-T_{j\,n+1\,k}=T_{jk\,n+1}=g^0_{jk}.
\end{equation}
Then, we have a skewsymmetric indexed family $T_{\lambda\mu\nu}$ with (greek)
indices running from $1$ to $n+1$, extending $T_{ijk}$ (recall that latin
indices run from $1$ to $n$). We have the following statement.

\begin{lemma}\label{lemma:projrec}
  A projective reciprocal transformation induces the transformation
  \begin{equation}
    \label{eq:623}
    T_{\lambda\mu\nu} = \frac{1}{A^3}\tilde{T}_{\alpha\beta\gamma}
    a^\alpha_\lambda a^\beta_\mu a^\gamma_\nu.
  \end{equation}
  Thus, $T_{\lambda\mu\nu}$ transforms as a tensor in
  $\mathbb{C}^{n+1}$ up to a conformal factor.
\end{lemma}
\begin{proof}
  It follows from
  \begin{align*}
    T_{lcs}= & \frac{1}{2A^3}\Big(\tilde{T}_{ijk}(a^i_la^j_c-a^i_ca^j_l)a^k_s
    +\tilde{T}_{ij\,n+1}(a^i_la^j_c - a^i_ca^j_l)a^{n+1}_s
    \\
    & -\tilde{T}_{ij\,n+1}(a^i_la^{n+1}_c-a^i_ca^{n+1}_l)a^j_s
      -\tilde{T}_{ij\,n+1}(a^{n+1}_la^j_c-a^{n+1}_ca^j_l)a^i_s\Big)
    \\
    = & \frac{1}{2A^3}\Big(\tilde{T}_{ij\nu}(a^i_la^j_c-a^i_ca^j_l)a^\nu_s
        +\tilde{T}_{i\,n+1\,k}(a^i_la^{n+1}_c-a^i_ca^{n+1}_l)a^k_s
    \\
    & +\tilde{T}_{n+1\,jk}(a^{n+1}_la^j_c-a^{n+1}_ca^j_l)a^k_s\Big)
    \\
    = & \frac{1}{2A^3}\tilde{T}_{\lambda\mu\nu}
        (a^\lambda_la^\mu_c-a^\lambda_ca^\mu_l)a^\nu_s
    \\
    = & \frac{1}{A^3}\tilde{T}_{\lambda\mu\nu}a^\lambda_la^\mu_ca^\nu_s.
  \end{align*}
  A similar proof holds for $T_{lc\,n+1}=g^{0}_{lc}$.
\end{proof}

\medskip

In what follows we will identify three-forms
$\omega\in\wedge^3(\mathbb{C}^{n+1})^*$ on a vector space $\mathbb{C}^{n+1}$
with maps of the form (see also \cite{DePoiFaenziMezzettiRanestad17} for more
details)
\begin{equation}
  \label{eq:20}
  i(\omega)\colon \mathbb{C}^{n+1} \to \wedge^2 (\mathbb{C}^{n+1})^*,
  v\mapsto \frac{1}{3}i_v(\omega).
\end{equation}
Clearly, the map $\omega\mapsto i(\omega)$ is an isomorphism onto its image.
If $(v^i)$ are coordinates on $\mathbb{C}^{n+1}$, then $(dv^i)$ is a basis of
$(\mathbb{C}^{n+1})^*$ and the above isomorphism reads as
\begin{equation}
  \label{eq:21}
  \omega_{ijk}dv^i\wedge dv^j\wedge dv^k \mapsto \omega_{ijk}v^k dv^i\wedge dv^j.
\end{equation}

\begin{theorem}\label{th:corresp}
  There is a bijective correspondence between leading coefficients of second
  order HHOs $g=(T_{ijk}u^k + g^0_{ij})du^i\wedge du^j$ as in \eqref{aa3}, and
  three-forms
  $\omega=\omega_{\lambda\mu\nu}dv^{\lambda}\wedge dv^{\mu}\wedge
  dv^{\nu}$. The bijective correspondence is preserved by projective reciprocal
  transformations up to a conformal factor.
\end{theorem}
\begin{proof}
  Let us consider a three-form
  $\omega=\omega_{\lambda\mu\nu}dv^{\lambda}\wedge dv^{\mu}\wedge
  dv^{\nu}$. Using the isomorphism~\eqref{eq:21} we can rewrite the form as
  \begin{align*}
    i(\omega) =& \omega_{\lambda\mu\nu}
    v^\nu dv^{\lambda}\wedge dv^{\mu}
    \\
    =& \omega_{i\mu\nu} v^\nu dv^{i}\wedge dv^{\mu} +
       \omega_{n+1\,\mu\nu} v^\nu dv^{n+1}\wedge dv^{\mu}
       \\
    & + \omega_{\lambda i\nu} v^\nu dv^\lambda\wedge dv^i +
      \omega_{\lambda\, n+1\nu} v^\nu dv^\lambda\wedge dv^{n+1}
    \\
    & + \omega_{\lambda\mu\,i} v^i dv^\lambda\wedge dv^\mu +
      \omega_{\lambda\mu\, n+1} v^{n+1} dv^\lambda\wedge dv^\mu
    \\
    =& \omega_{ij\nu} v^\nu dv^{i}\wedge dv^j +
       \omega_{i\,n+1\,j} v^j dv^{i}\wedge dv^{n+1} +
       \omega_{n+1\,ij} v^j dv^{n+1}\wedge dv^i
    \\
    & + \omega_{i j\nu} v^\nu dv^i\wedge dv^j
      + \omega_{n+1\, ij} v^j dv^{n+1}\wedge dv^i
      + \omega_{i\, n+1\,j} v^j dv^i\wedge dv^{n+1}
    \\
    & + \omega_{j\mu i} v^i dv^j\wedge dv^\mu +
      \omega_{n+1\,ij} v^j dv^{n+1}\wedge dv^i +
      \omega_{ij\, n+1} v^{n+1} dv^i\wedge dv^j
    \\
    = & 3\omega_{ijk}v^kdv^i\wedge dv^j + 3\omega_{ij\,n+1}v^{n+1}dv^i\wedge dv^j
        + 6 \omega_{ij\,n+1} v^idv^j\wedge dv^{n+1}.
  \end{align*}
  Using the affine chart restriction $v^{n+1}=1$, $dv^{n+1}=0$ we obtain a
  second-order HHO by setting
  \begin{equation}
    \label{eq:17}
    T_{ijk} = 3\omega_{ijk}\quad\text{and} \quad g_{ij}^0 = T_{ij\,n+1} =
    3\omega_{ij\,n+1}.
  \end{equation}

  On the other hand, from a second-order HHO $g$ as in the statement one can
  define the form in homogeneous coordinates
  \begin{equation}
    \label{eq:19}
    G= (T_{ijk}v^k + g^0_{ij}v^{n+1})dv^i\wedge dv^j.
  \end{equation}
  Reversing the steps of the first part of the proof we get the desired
  three-form $\omega$.

  The fact that the correspondence is preserved by projective reciprocal
  transformation up to the conformal factor $1/A^3$ follows from
  Lemma~\ref{lemma:projrec}.
\end{proof}

There is an immediate and important consequence of the above Theorem.

\begin{corollary}
  There is a bijective correspondence between homogeneous second order
  Hamiltonian operators in dimension $n$ and three-forms in dimension $n+1$. 
  The bijective correspondence is preserved by projective reciprocal
  transformations.
\end{corollary}

At this point we observe two important facts:
\begin{itemize}
\item from a \emph{geometric} viewpoint, second-order HHOs yield algebraic
  varieties using the corresponding three-forms and the mechanism explained in
  \cite{DePoiFaenziMezzettiRanestad17}.
\item from an \emph{algebraic} viewpoint, second-order HHOs can be classified
  under the action of the projective reciprocal transformations
  by means of the classification of three-forms under the action of
  $SL(n+1,\mathbb{C})$.
\end{itemize}

Let us first summarize the main features of the geometric properties of
second-order HHOs. Our main source is~\cite{DePoiFaenziMezzettiRanestad17}.
Let $\omega$ be a three-form as above.  A line $L$ in $\mathbb{C}^{n+1}$ can be
identified with the skew-symmetric tensor
$L=p^{\lambda\mu}\pd{}{v^\lambda}\wedge\pd{}{v^\mu}$; $(p^{\lambda\mu})$ are
the Pl\"ucker coordinates. The system
\begin{equation}
  \label{eq:16}
  i_L\omega = 0 \Leftrightarrow \omega_{\lambda\mu\nu}p^{\mu\nu} = 0
\end{equation}
is a system of $n+1$ linear equations whose solutions constitute a linear
subspace $\Lambda_{\omega}\subset\mathbb{P}(\wedge^2\mathbb{C}^{n+1})$. If
$\omega$ is a generic $3$-form, then the intersection of $\Lambda_{\omega}$
with the Grassmannian $\mathbb{G}$, $X_{\omega}=\mathbb{G}\cap \Lambda_\omega$
is an $n-1$-dimensional variety, \emph{i.e.} it is a \emph{linear line
  congruence}.

As a direct consequence of Theorem~\ref{th:corresp}, the problem of classifying
non-degenerate $n$-components second-order HHOs under the action of projective
reciprocal transformations is solved by means of the the classification of
$3$-forms in $\mathbb{C}^{n+1}$ under the action of the group
$SL(n+1,\mathbb{C})$. This is what will be exposed in next Section.

\subsection{Projective classification of Hamiltonian operators}
\label{sec:proj-class-hamilt-1}

The following results are a direct consequence of the classification of
$3$-forms in $\mathbb{C}^{n+1}$ under the action of the group
$SL(n+1,\mathbb{C})$. Such a classification can be found in the book
\cite{gurevich64:_invariants} for $n\leq 7$, while the case $n=8$ is covered in
\cite{vinberg88:_class_trivec_dimen_space}. It should be remarked that the
latter paper presents the classification of trivectors in dimension $9$,
\emph{i.e.} elements of $\wedge^3\mathbb{C}^{9}$, under the natural action of
$SL(n+1,\mathbb{C})$. It is easy to realize that the classification of
$3$-forms (\emph{i.e.} the set of orbits) is put in bijective correspondence
with the classification of trivectors by any isomorphism, for example, the
correspondence defined by the passage from a basis to its dual
$e_i\mapsto e^i$.

\paragraph{The case $n=2$.}  There is only one (nontrivial) $3$-form, namely
$\omega=dv^1\wedge dv^2\wedge dv^3$. We can rewrite it as
\begin{equation}
  \label{eq:25}
  i(\omega) = \frac{1}{3}(v^1dv^2\wedge dv^3 - v^2dv^3\wedge dv^1
  + v^3dv^1\wedge dv^2)
\end{equation}
The affine projection $v^3=1$, $dv^3=0$, yields, up to a factor, the leading
coefficient $du^1\wedge du^2$ of the second-order HHO
\begin{equation}
  \label{eq:18}
  R = \begin{pmatrix}0 & 1 \\ -1 & 0
  \end{pmatrix}\partial_x^2.
\end{equation}

\paragraph{The case $n=4$.}  There are two (nontrivial) orbits. The open orbit
is generated by
\begin{equation}
  \label{eq:22}
  \omega = dv^5\wedge(dv^1\wedge dv^2 + dv^3\wedge dv^4),
\end{equation}
that corresponds to the leading coefficient $du^1\wedge du^2 + du^3\wedge du^4$
of the second-order HHO
\begin{equation}
  \label{eq:26}
  R = \begin{pmatrix}0 & 1 & 0 & 0\\ -1 & 0 & 0 & 0\\
    0 & 0 & 0 & 1\\ 0 & 0 & -1 & 0
\end{pmatrix}\partial_x^2.
\end{equation}
the closed orbit is totally decomposable and generated by
\begin{equation}
  \label{eq:23}
  \omega = dv^1\wedge dv^2 \wedge dv^3;
\end{equation}
the corresponding leading coefficient is degenerate: $\det(g_{ij})=0$.

\paragraph{The case $n=6$.} The classification in this case is due to Schouten
(see \cite{gurevich64:_invariants}). There are nine nontrivial orbits. We list
below the generators of the orbits which lead to a non-degenerate $2$-form
$i(\omega)$.
  \begin{enumerate}
  \item The open orbit is generated by
    \begin{multline}
      \label{eq:24}
      \omega = dv^1\wedge dv^2 \wedge dv^3 + dv^4\wedge dv^5 \wedge dv^6
      \\
      + dv^7\wedge(dv^1\wedge dv^4 + dv^2\wedge dv^5 + dv^3\wedge dv^6).
    \end{multline}
    (case X in \cite{gurevich64:_invariants}). By using the map $i(\omega)$:
    \begin{align}
      i(\omega)&=\frac{1}{3}(v^1dv^2\wedge dv^3-v^2dv^1\wedge dv^3
                 +v^3dv^1\wedge dv^2+\\
               &\hphantom{ciao}v^4dv^5\wedge dv^6-v^5dv^4\wedge dv^6
                 +v^6dv^4\wedge dv^5+\\
               &\hphantom{ciao}v^7dv^1\wedge dv^4-v^1dv^7\wedge dv^4
                 +v^4dv^7\wedge dv^1+\\
               &\hphantom{ciao}v^7dv^2\wedge dv^5-v^2dv^7\wedge dv^5
                 +v^5dv^7\wedge dv^2+\\
               &\hphantom{ciao}v^7dv^3\wedge dv^6-v^3dv^7\wedge dv^6
                 +v^6dv^7\wedge dv^3)
    \end{align}
    Then with the affine projection $v^7=1, dv^7=0$:
    \begin{align}
      i(\omega)=&\frac{1}{3}(v^3dv^1\wedge dv^2-v^2dv^1\wedge dv^3
                  +v^3dv^1\wedge dv^2+\\
                &\hphantom{ci}v^4dv^5\wedge dv^6-v^5dv^4\wedge dv^6
                  +v^6dv^4\wedge dv^5+\\
                &\hphantom{ci}dv^1\wedge dv^4+dv^2\wedge dv^5+dv^3\wedge dv^6)
    \end{align}
    Then, the associated 2-form is (up to a factor)
    \begin{equation}
      g_{ij}^{1}=\begin{pmatrix}
        0&v^3&-v^2&1&0&0\\
        -v^3&0&v^1&0&1&0\\
        v^2&-v^1&0&0&0&1\\
        -1&0&0&0&v^6&-v^5\\
        0&-1&0&-v^6&0&v^4\\
        0&0&-1&v^5&-v^4&0
      \end{pmatrix}
    \end{equation}
    and $\text{det}(g^{1}_{ij})=(v^1v^4+v^2v^5+v^3v^6-1)^2$.

\item We have the $3$-form
\begin{multline}
  \omega = dv^1\wedge dv^2\wedge dv^3+dv^4\wedge dv^5\wedge dv^6
  \\
  +(dv^1\wedge dv^4+dv^2\wedge dv^5)\wedge dv^7
\end{multline}
(case IX in \cite{gurevich64:_invariants}). In the affine chart (removing the
factor $1/3$),
\begin{align}
  3i(\omega)&=v^1dv^2\wedge dv^3-v^2dv^1\wedge dv^3
             +v^3dv^1\wedge dv^2+\\
&\hphantom{ciao}v^4dv^5\wedge dv^6-v^5dv^4\wedge dv^6+v^6dv^4\wedge dv^5+\\
&\hphantom{ciao}dv^1\wedge dv^4+dv^2\wedge dv^5
\end{align}
The leading coefficient of the associated operator is
\begin{equation}
g_{ij}^2=\begin{pmatrix}
0&v^3&-v^2&1&0&0\\
-v^3&0&v^1&0&1&0\\
v^2&-v^1&0&0&0&0\\
-1&0&0&0&v^6&-v^5\\
0&-1&0&-v^6&0&v^4\\
0&0&0&v^5&-v^4&0
\end{pmatrix}
\end{equation}
we have $\text{det}(g^2_{ij})=(v^1v^4+v^2v^5)^2$.

\item We have the $3$-form
\begin{equation}
\omega = dv^1\wedge dv^2\wedge dv^3+dv^4\wedge dv^5\wedge dv^6
+dv^1\wedge dv^4\wedge dv^7
\end{equation}
(case VIII in \cite{gurevich64:_invariants}). In the affine chart,
\begin{multline}
  3i(\omega)=v^1dv^2\wedge dv^3-v^2dv^1\wedge dv^3
             +v^3dv^1\wedge dv^2
  \\
  +v^4dv^5\wedge dv^6-v^5dv^4\wedge dv^6+v^6dv^4\wedge dv^5+dv^1\wedge dv^4
\end{multline}
The leading coefficient of the associated operator is
\begin{equation}
g_{ij}^3=\begin{pmatrix}
0&v^3&-v^2&1&0&0\\
-v^3&0&v^1&0&0&0\\
v^2&-v^1&0&0&0&0\\
-1&0&0&0&v^6&-v^5\\
0&0&0&-v^6&0&v^4\\
0&0&0&v^5&-v^4&0
\end{pmatrix}
\end{equation}
we have $\text{det}(g^3_{ij})=(v^1v^4)^2$.

\item We have the $3$-form
\begin{equation}
  \omega = dv^4\wedge dv^5\wedge dv^6+dv^7(du^1\wedge dv^4
  +dv^2\wedge dv^5+dv^3\wedge dv^6)
\end{equation}
(case VII in \cite{gurevich64:_invariants}). In the affine chart,
\begin{multline}
  3i(\omega)=v^4dv^5\wedge dv^6-v^5dv^4\wedge dv^6+v^6dv^4\wedge dv^5
  \\
  +dv^1\wedge dv^4+dv^2\wedge dv^5+dv^3\wedge dv^6
\end{multline}
The leading coefficient of the associated operator is
\begin{equation}
g_{ij}^4=\begin{pmatrix}
0&0&0&1&0&0\\
0&0&0&0&1&0\\
0&0&0&0&0&1\\
-1&0&0&0&v^6&-v^5\\
0&-1&0&-v^6&0&v^4\\
0&0&-1&v^5&-v^4&0
\end{pmatrix}
\end{equation}
we have $\text{det}(g^4_{ij})=1$.

\item  We have the $3$-form
\begin{equation}
\omega = dv^7\wedge(dv^1\wedge dv^4+dv^2\wedge dv^5+dv^3\wedge dv^6)
\end{equation}
(case VI in \cite{gurevich64:_invariants}).
In the affine chart we have,
\begin{equation}
3i(\omega)=dv^1\wedge dv^4+dv^2\wedge dv^5+dv^3\wedge dv^6
\end{equation}
The leading coefficient of the associated operator is
\begin{equation}
g_{ij}^6=\begin{pmatrix}
0&0&0&1&0&0\\
0&0&0&0&1&0\\
0&0&0&0&0&1\\
-1&0&0&0&0&0\\
0&-1&0&0&0&0\\
0&0&-1&0&0&0
\end{pmatrix}
\end{equation}
we have $\text{det}(g^6_{ij})=1$.
\end{enumerate}

\paragraph{The case $n=8$.} We will follow the classification of trivectors in
dimension $9$ \cite{vinberg88:_class_trivec_dimen_space}. We will use the
isomorphism between $\mathbb{C}^{n+1}$ and $(\mathbb{C}^{n+1})^*$ defined by a
basis and its dual in order to put trivectors and $3$-forms into
correspondence. We recall that a trivector is said to be \emph{semisimple} if
its equivalence class is closed in the space of all trivectors, whereas it is
said to be \emph{nilpotent} if the closure of this class contains the zero
form.  Every trivector $u$ can be uniquely written as the sum $u=p+e$, where
$p$ is a semisimple trivector and $e$ is a nilpotent trivector such that
$p\wedge e=0$.

Semisimple trivectors $p$ are divided into seven different families for each of
which all possible nilpotent parts are provided.  Let us introduce the
following $3$-forms:
\begin{align}
p_1&=dv^1\wedge dv^2\wedge dv^3+dv^4\wedge dv^5\wedge dv^6+dv^7\wedge dv^8\wedge dv^9\\
p_2&=dv^1\wedge dv^4\wedge dv^7+dv^2\wedge dv^5\wedge dv^8+dv^3\wedge dv^6\wedge dv^9\\
p_3&=dv^1\wedge dv^5\wedge dv^9+dv^2\wedge dv^6 \wedge dv^7+dv^3\wedge dv^4\wedge dv^8\\
p_4&=dv^1\wedge dv^6\wedge dv^8+dv^2\wedge dv^4\wedge dv^9+dv^3\wedge dv^5\wedge dv^7
\end{align}

Every semisimple trivector is equivalent to a trivector whose corresponding
$3$-form is of the type
\begin{equation}\label{eq:28}
p=\lambda_1p_1+\lambda_2p_2+\lambda_3p_3+\lambda_4p_4,
\end{equation}
where the coefficients are determined up to a linear transformation from a
group generated by complex reflections of order $3$
\cite{vinberg88:_class_trivec_dimen_space}.

The first family of $3$-forms is generated by $p$ only as in~\eqref{eq:28};
more precisely, it consists only of semisimple trivectors ($e=0$).  The
coefficients $\lambda_i$ must satisfy a complicated system of algebraic
inequalities \cite{vinberg88:_class_trivec_dimen_space}.
The stabilizer subgroup $S$ of this class is a cyclic Abelian group of order
81.  The corresponding non-degenerate $2$-form in this class is
\begin{equation}
  g_{ij}^{(1)}
  =  \begin {pmatrix}  0&\lambda_1 v^3&-\lambda_1 v^2&
\lambda_2 v^7&\lambda_3&\lambda_4 v^8&-\lambda_2 v^4&-\lambda_4 v^6
\\  { }-\lambda_1 v^3&0&\lambda_1 v^1&\lambda_4&\lambda_2
 v^8&\lambda_3 v^7&-\lambda_3 v^6&-\lambda_2 v^5\\  { }
\lambda_1 v^2&-\lambda_1 v^1&0&\lambda_3 v^8&\lambda_4 v^7&\lambda_2&-
\lambda_4 v^5&-\lambda_3 v^4\\  { }-\lambda_2 v^7&-
\lambda_4&-\lambda_3 v^8&0&\lambda_1 v^6&-\lambda_1 v^5&\lambda_2 v^1&
\lambda_3 v^3\\  { }-\lambda_3&-\lambda_2 v^8&-\lambda_4
v^7&-\lambda_1 v^6&0&\lambda_1 v^4&\lambda_4 v^3&\lambda_2 v^2
\\  { }-\lambda_4 v^8&-\lambda_3 v^7&-\lambda_2&\lambda_1
 v^5&-\lambda_1 v^4&0&\lambda_3 v^2&\lambda_4 v^1
\\  { }\lambda_2 v^4&\lambda_3 v^6&\lambda_4 v^5&-
\lambda_2 v^1&-\lambda_4 v^3&-\lambda_3 v^2&0&\lambda_1
\\  { }\lambda_4 v^6&\lambda_2 v^5&\lambda_3 v^4&-
\lambda_3 v^3&-\lambda_2 v^2&-\lambda_4 v^1&-\lambda_1&0\end {pmatrix}
\end{equation}

The second family is generated by the semisimple trivector
\begin{equation}
p=\lambda_1p_1+\lambda_2p_2-\lambda_3p_3,
\end{equation}
again with $\lambda_i$ fulfilling an algebraic constraint.  The coefficients
are determined up to a linear transformation generated by complex reflections.
The possible nontrivial nilpotent parts are two:
\begin{align}
&e_1=dv^1\wedge dv^6\wedge dv^8+dv^2\wedge dv^4\wedge dv^9\\
&e_2=dv^1\wedge dv^6\wedge dv^8
\end{align}
Here, the dimension of the stabilizer $S$ is $0$ for $e_1$ and $1$ for $e_2$.
By summing $p+e_i$ and applying the correspondence, we finally obtain the
following two $2$-forms:
\begin{equation} \footnotesize
  g_{ij}^{(2)}= \begin {pmatrix}  0&\lambda_1 u^3&-\lambda_1 u^2&
\lambda_2 u^7&-\lambda_3&u^8&-\lambda_2 u^4&-u^6\\  { }-
\lambda_1 u^3&0&\lambda_1 u^1&1&\lambda_2 u^8&-\lambda_3 v^7&\lambda_3
 u^6&-\lambda_2 u^5\\  { }\lambda_1 u^2&-\lambda_1 u^1&0
&-\lambda_3 u^8&0&\lambda_2&0&\lambda_3 u^4\\  { }-
\lambda_2 v^7&-1&\lambda_3 u^8&0&\lambda_1 u^6&-\lambda_1 u^5&\lambda_2
 u^1&-\lambda_3 u^3\\  { }\lambda_3&-\lambda_2 u^8&0&-
\lambda_1 u^6&0&\lambda_1 u^4&0&\lambda_2 u^2\\  { }-u^8
&\lambda_3 v^7&-\lambda_2&\lambda_1 u^5&-\lambda_1 u^4&0&-\lambda_3 u^2
&u^1\\  { }\lambda_2 u^4&-\lambda_3 u^6&0&-\lambda_2 u^1
&0&\lambda_3 u^2&0&\lambda_1\\  { }u^6&\lambda_2 u^5&-
\lambda_3 u^4&\lambda_3 u^3&-\lambda_2 u^2&-u^1&-\lambda_1&0
\end {pmatrix}
\end{equation}
\begin{equation}\footnotesize
g_{ij}^{(3)}=   \begin {pmatrix}  0&\lambda_1 u^3&-\lambda_1 u^2&
\lambda_2 v^7&-\lambda_3&u^8&-\lambda_2 u^4&-u^6\\  { }-
\lambda_1 u^3&0&\lambda_1 u^1&0&\lambda_2 u^8&-\lambda_3 v^7&\lambda_3
 u^6&-\lambda_2 u^5\\  { }\lambda_1 u^2&-\lambda_1 u^1&0
&-\lambda_3 u^8&0&\lambda_2&0&\lambda_3 u^4\\  { }-
\lambda_2 v^7&0&\lambda_3 u^8&0&\lambda_1 u^6&-\lambda_1 u^5&\lambda_2
 u^1&-\lambda_3 u^3\\  { }\lambda_3&-\lambda_2 u^8&0&-
\lambda_1 u^6&0&\lambda_1 u^4&0&\lambda_2 u^2\\  { }-u^8
&\lambda_3 v^7&-\lambda_2&\lambda_1 u^5&-\lambda_1 u^4&0&-\lambda_3 u^2
&u^1\\  { }\lambda_2 u^4&-\lambda_3 u^6&0&-\lambda_2 u^1
&0&\lambda_3 u^2&0&\lambda_1\\  { }u^6&\lambda_2 u^5&-
\lambda_3 u^4&\lambda_3 u^3&-\lambda_2 u^2&-u^1&-\lambda_1&0
\end {pmatrix}
\end{equation}
In both cases the determinants are non-zero.

The total number of non-degenerate two-forms in the classification is $132$. For
reasons of space, we will not list elements in the families $3$--$7$; however,
we are ready to privately provide the list of non-degenerate two-forms to the
interested reader.

\begin{remark}
  The dimension of the space of $3$-forms $\wedge^3(\mathbb{C}^{n+1})^*$ grows
  with the dimension $n$ in a much faster way than the dimension of
  $SL(\mathbb{C}^{n+1})$. However, for small values of $n$ the dimension of the
  group is prevailing: this is the reason for the lack of non-trivial classes
  when $n\leq 4$. The same argument shows that a classification for higher
  values of $n$ does not make sense, in view of the huge number of free
  parameters that the generic element would depend on.
\end{remark}

\section{Systems of PDEs with second-order Hamiltonian
  structure}
\label{sec:second-order-hhos}

In \cite[Theorem 10]{vergallo20:_homog_hamil} it was proved that the necessary
conditions for a second-order HHO $P$~\eqref{1b} to be a Hamiltonian operator
for a quasilinear system of first-order conservation laws~\eqref{2a} are
\begin{subequations}\label{eq:37}
  \begin{gather}
    \label{eq:451}
    g_{qj}V^j_{,p} + g_{pj}V^j_{,q} = 0,
    \\
    \label{eq:38}
    g_{qk}V^k_{,pl} + g_{pq,k}V^k_{,l} + g_{qk,l}V^{k}_{,p}= 0.
  \end{gather}
\end{subequations}
This result is analogue to the results in \cite{tsarev85:_poiss_hamil}
concerning first-order HHOs and quasilinear systems of first-order PDEs, and is
obtained by a method introduced in \cite{KerstenKrasilshchikVerbovetsky:HOpC}
and later adapted to HHOs~\cite{FPV17:_system_cl}.

We observe that the above conditions have no direct differential-geometric
interpretation as they are derived in flat coordinates of the connection
$\Gamma^i_{jk}$ (see the Introduction). However, we will be able to parametrize
the space of solutions of the above equations, thus exhibiting large families
of systems of PDEs that are Hamiltonian with respect to second-order
Hamiltonian operators. Interesting properties of such systems will be
thoroughly investigated.

\subsection{Solution of the compatibility conditions}
\label{sec:hamiltonian-systems}

We will now solve completely the system of compatibility conditions between a
quasilinear system of first-order PDEs~\eqref{eq:37}. We will first prove that
the system is in involution, then we will parametrize its solutions.
\begin{proposition}\label{pr:solut-comp-cond}
  The system \eqref{eq:37} is in involution.
\end{proposition}
\begin{proof}
  Let us derive \eqref{eq:4}:
  \begin{equation}\label{aa1}
    g_{qj,l}V^j_{,p}+g_{qj}V^j_{,pl}+g_{pj,l}V^j_{,q}+g_{pj}V^j_{,ql}=0
  \end{equation}
  then, by using condition \eqref{eq:38} we can substitute
  \begin{gather}
    g_{qj}V^j_{,pl}+g_{qj,l}V^j_{,p}=-g_{pq,j}V^j_{,l}\\
    g_{pj}V^j_{,ql}+g_{pj,l}V^j_{,q}=-g_{qp,j}V^k_{,l}
  \end{gather} in \eqref{aa1}, which yields
  \begin{equation}
    -g_{pq,j}V^j_{,l}-g_{qp,j}V^j_{,l} =0,
  \end{equation}
  which is an identity.

  The condition \eqref{aa1} can be rewritten as
  \begin{equation}
    \label{eq:10}
    (g_{qk}V^k_{,p})_{,l} + g_{pq,k}V^k_{,l} = 0.
  \end{equation}
  From the consistency condition $V^k_{,plm}=V^k_{,pml}$ we obtain
  \begin{equation}
    \label{eq:11}
    (g_{qk}V^k_{,p})_{,lm} + (g_{pq,k}V^k_{,l})_{,m} =
    \\
    (g_{qk}V^k_{,p})_{,ml} + (g_{pq,k}V^k_{,m})_{,l}
  \end{equation}
  which yields the identity $g_{pq,k}V^k_{,lm} = g_{pq,k}V^k_{,ml}$
  in view of $g_{pq,kl}=0$.
\end{proof}

\medskip

The above Proposition shows that, since \eqref{aa1} expresses all second-order
derivatives, the general solution of the system depends on no more than $n+n^2$
parameters. The equations \eqref{eq:4} impose further $n(n-1)/2$ additional
restrictions, so that the total number of arbitrary constants in the general
solution is
\begin{equation}\label{eq:12}
n+n^2 - \frac{n(n-1)}{2} = n(n+3)/2.
\end{equation}

Now, we will solve the system \eqref{eq:37}.

\begin{theorem}\label{th:solut-comp-cond}
  Let $C$ be a second-order HHO. Then, the (explicit) solution of the
  system~\eqref{eq:37} is the vector $V^i$ given by
  \begin{equation}
    V^i = g^{ij}W_j,
  \end{equation}
  where $W_j$ is the covector
  \begin{equation}
    W_j=A_{jl}u^l+B_j
  \end{equation}
  where $A_{ij}=-A_{ji}$, $B_i$ are arbitrary constants.
\end{theorem}
\begin{proof}
  \begin{equation}\label{ght}
    (W_j)_{,ab}=g_{jk}V^k_{,ab}+g_{jk,b}V^k_{,a}+g_{jk,ab}V^k+g_{jk,a}V^k_{,b}
  \end{equation}
  Since $g_{jk,ab}=0$ and $g_{bj,k}= g_{jk,b}$ the above equation becomes the
  right-hand side of~\eqref{eq:38}. Then $W_j=A_{jl}u^l+B_j$. Moreover, we have
  the following identity:
  \begin{equation}
    \label{eq:1}
    W_{j,p} + W_{p,j} = g_{jl}V^l_{,p} + g_{pl}V^l_{,j}
  \end{equation}
  which is the right-hand side of~\eqref{eq:4}. But we have
  \begin{equation}\label{eq:2}
    W_{j,p}+W_{p,j} = A_{jp}+A_{pj},
  \end{equation}
  which completes the proof.
\end{proof}

\begin{remark}
  The above solution of the system~\eqref{eq:37} is the most general: indeed,
  $W_i$ depends on $n(n+1)/2$ arbitrary constants, and $V^i$ is defined up to
  $n$ arbitrary constants (as it enters the right-hand side of \eqref{2a}). The
  total figure is equal to the dimension count following the proof of
  Proposition~\ref{pr:solut-comp-cond}.
\end{remark}

\begin{corollary}
  The fluxes $V^i$ are rational functions of the form:
  \begin{equation}
    V^i=\frac{Q}{\operatorname{Pf}(g)}
  \end{equation}
  where $Q$ is a polynomial of degree $n/2$ and the denominator is
  $\operatorname{Pf}(g)$.
\end{corollary}
\begin{proof}
  We have $V^i=g^{is}W_s$, where $g^{ij}$ is the inverse matrix of $g_{ij}$. By
  means of properties of the determinant of skew-symmetric matrices \cite{Muir:TtD60}
  the inverse matrix has rational functions entries where the numerator has
  degree $(n-2)/2$ and the denominator is the Pfaffian of $g_{ij}$, whose
  degree is at most $n/2$. Since $W_s$ are linear functions, the statement
  follows.
\end{proof}

\begin{corollary}\label{cor:eigenvalues}
  The eigenvalues of the matrix $V^i_j$ have algebraic multiplicity $2$.
\end{corollary}
\begin{proof}
The eigenvalues are computed by the characteristic polynomial:
$\det(V^i_j-\lambda\delta^i_j)$. Lowering one index we obtain a skew-symmetric
matrix:

\begin{align*}
  g_{hi}(V^i_j-\lambda\delta^i_j)=&g_{hi}(g^{ik}W_k)_{,j}-\lambda g_{hj}
  \\ &=-g_{hi,j}g^{ik}W_k+g_{hi}g^{ik}A_{kj}-\lambda g_{hj}
  \\ &=  T_{hji}g^{ik}W_k+A_{hj}-\lambda g_{hj}
\end{align*}
whose determinant is the square of its Pfaffian. Since $\det(g_{ij})$ is also a
perfect square we get the result.
\end{proof}

\begin{remark}
  The above result has important consequences on the integrability of the
  system \eqref{2a}. Indeed, in \cite{tsarev91:_hamil} the generalized
  hodograph method for the solution of semi-Hamiltonian quasilinear first-order
  systems is developed. However, one of the hypothesis in the above paper is
  that all eigenvalues of $V^i_j$ are distinct. We will see in
  Subsection~\ref{sec:haantj-tens-integr} that, at least in some cases,
  this does not prevent the semi-Hamiltonianity of the systems determined by
  second-order HHOs.
\end{remark}

\subsection{Projective geometry of the systems of PDEs}
\label{sec:proj-geom-hamilt}

Let us recall that for every hydrodynamic-type system
\begin{equation}
  u^i_t=V^i_j(u)u^j_x
\end{equation}
it is possible to associate a congruence
\begin{equation}
  y^i=u^iy^{n+1}+V^iy^{n+2}
\end{equation}
in auxiliary projective space $\mathbb{P}^{n+1}$ with homogeneous coordinates
$(y^1:\cdots : y^{n+2})$.
\begin{proposition}
  Let $u^i_t=V^i_ju^j_x$ be a system compatible with a second order HHO. Then
  the associated congruence is linear.
\end{proposition}
\begin{proof}
  By the previous Lemma we obtain that
  \begin{equation}
    A_{jl}u^l+B_j=(T_{jkl}u^l+g^0_{jk})V^k
  \end{equation}
  then
    \begin{equation}
      \frac{1}{2}T_{jkl}\left(u^lV^k-u^kV^l\right)+g^0_{jk}V^k=A_{jl}u^l+B_j
    \end{equation}
  This yields a system of $n$ linear relations between Pl\"ucker's coordinates
  describing the line congruence, hence the statement is proved.
\end{proof}

\begin{corollary}\label{cor:lindeg}
  The first-order quasilinear systems of PDEs compatible with second-order HHOs
  are linearly degenerate.
\end{corollary}
\begin{proof}
  It is a straightforward consequence of the linearity of the corresponding
  congruence of lines \cite{agafonov96:_system}.
\end{proof}

\begin{remark}
  The fact that the associated congruence is linear does not imply that the
  system is non-diagonalizable (see the Remark in
  \cite[p.\ 1771]{agafonov99:_theor}).
\end{remark}

We can have a look at the case $n=2$ (already considered in
\cite{vergallo20:_homog_hamil}) using the theory of congruences.

\begin{proposition}
  For $n=2$ the system $V^i_j$ compatible with a second order HHO is
  linearisable.
\end{proposition}
\begin{proof}Let $n=2$. Then in $\mathbb{P}^3$ linear congruences can be
  brought (modulo projective transformations) to the form
  \begin{equation}
    y^1=u^2y^3+u^2y^3\qquad y^2=u^2y^3+u^1y^4
  \end{equation}
  By using the affine chart $y^4=1$ we have only two cases $(y^1=y^2, y^3=1)$
  or $(y^1=-y^2,y^3=-1)$. But by condition \eqref{eq:451} we obtain that the
  system is skew symmetric, then
  \begin{equation}
    \begin{cases}u^1_t=u^2_x\\u^2_t=-u^1_x\end{cases}
  \end{equation}
  In particular, every system $u^i_t=(V^i)_x$ compatible with a second order
  operator is linearisable.
\end{proof}

\medskip

We stress that an obvious alternative proof immediately follows from the
classification of second-order HHOs in Subection~\ref{sec:proj-class-hamilt-1}.
The same argument yields the following Proposition.

\begin{proposition}
  In the case $n=4$ systems of quasilinear first-order conservation laws
  that admit a second-order HHO are linearisable by projective reciprocal
  transformations.
\end{proposition}

\subsection{Hamiltonian systems}
\label{sec:hamiltonian-systems-1}

We would like to find an Hamiltonian for the systems that we found in
Section~\ref{sec:hamiltonian-systems}. To this aim we observe that, in
potential coordinates $u^i=b^i_x$, the second-order HHOs~\eqref{1b} in potential
coordinates becomes the ultralocal operator $P^{ij}= - g^{ij}(b^k_x)$, and the
system of first-order conservation laws \eqref{2a} becomes $b^i_t = V^i(b^k_x)$.
Solving with respect to $H$ the equation:
\begin{equation}
  - g^{ik}\frac{\delta H}{\delta b^k}=V^i,
\end{equation}
or
\begin{equation}
  \frac{\delta H}{\delta b^k}= - A_{ks}b_x^s-B_k,
\end{equation}
yields the following result.

\begin{proposition}\label{pr:hamiltonian-systems-2}
  We have
  \begin{equation}
    H=-\int{\left(\frac{1}{2}A_{sl}b^l_x+B_s\right)b^sdx}
  \end{equation}
\end{proposition}
\begin{remark}
  At difference with the third-order case \cite{FPV17:_system_cl} we observe
  that there are no non-trivial nonlocal Casimirs, as the equation
  \begin{equation}
    -g^{ik}\frac{\delta F^j}{\delta b^k}=0
  \end{equation}
  has no non-trivial solutions.
\end{remark}

\subsection{The Haantjes tensor and integrability}
\label{sec:haantj-tens-integr}

For general quasilinear first-order systems of the type
$u^i_t = V^i_j(u^k)u^j_x$ it is known~\cite{haantjes55:_x} that the Haantjes
tensor
\begin{equation}
  H^i_{jk}=N^i_{pr}V^p_jV^r_k-N^p_{jr}V^i_pV^r_k
  - N^p_{rk}V^i_pV^r_j + N^p_{jk}V^i_rV^r_p,
\end{equation}
where $N^i_{jk}$ is the Nijenhuis tensor
\begin{equation}\label{nij}
  N^i_{jk}=V^p_jV^i_{kp}-V^p_{k}V^i_{jp}-V^i_p(V^p_{kj}-V^p_{jk}),
\end{equation}
is related with the diagonalizability of the quasilinear system of PDEs
\eqref{2a}. More precisely, if $H^i_{jk}=0$ and the eigenvalues of the velocity
matrix $V^i_{,j}$ have the same algebraic multiplicity as their geometric
multiplicity, then the system of PDEs is diagonalizable.

It is possible to compute the Haantjes tensor for our conservative
systems~\eqref{2a}. We have the following result.
\begin{theorem}\label{haant}
  The Haantjes tensor of a conservative quasilinear system~\eqref{2a} that
  admits a second-order homogeneous Hamiltonian operator is identically
  vanishing.
\end{theorem}
\begin{proof}
  It is easy to prove the following identity
  \begin{equation}\label{eq:5}
    g^{ia}_{,k}g_{aj}=-g^{ia}T_{ajk}
  \end{equation}
  Then, from \eqref{eq:38} we have $V^a_{pl}=-g^{aq}(T_{pqk}V^k_{,l}
  +T_{qkl}V^k_{,p})$, hence the Nijenhuis tensor can be written as
  \begin{equation}\label{eq:6}
    N^i_{jk}=g^{ia}\left(T_{jal}V^l_pV^p_k-T_{kal}V^l_pV^p_j
      -2T_{alp}V^l_kV^p_j\right).
  \end{equation}
  Using~\eqref{eq:6} we obtain
  \begin{equation}
    \label{han1}
    \begin{split}
        H^i_{kj}=&g^{ia}\left(T_{pal}V^l_{s}V^s_r-T_{ral}V^l_sV^s_p
      -2T_{als}V^l_rV^s_p\right)V^p_kV^r_j+\\
    &-g^{pa}\left(T_{kal}V^l_sV^s_r-T_{ral}V^l_sV^s_k
      -2T_{als}V^l_rV^s_k\right)V^i_pV^r_j+\\
    &-g^{pa}\left(T_{ral}V^l_sV^s_j-T_{jal}V^l_sV^s_r
      -2T_{als}V^l_jV^s_r\right)V^i_pV^r_k+\\
    &+g^{pa}\left(T_{kal}V^l_sV^s_j-T_{jal}V^l_sV^s_k
      -2T_{als}V^l_jV^s_k\right)V^i_rV^r_p
  \end{split}
\end{equation}
Let us consider, for example, the summand
\begin{equation}\label{Ss1}
  S=-2g^{ia}T_{als}V^l_rV^s_pV^p_kV^r_j+2g^{pa}T_{als}V^l_rV^s_kV^i_pV^r_j.
\end{equation}
It is clear that $S=0$ if $-2g^{ia}T_{als}V^s_pV^p_k +
2g^{pa}T_{als}V^s_kV^i_p=0$. Now, we use the identities~\eqref{eq:5},
\eqref{eq:4} and the upper indices version $g^{il}V^j_{,l} + g^{jl}V^i_{,l} =
0$ to prove that $g^{ia}T_{als}V^s_pV^p_k = g^{pa}T_{als}V^s_kV^i_p$, so that
$S=0$.

With similar algebraic manipulations it is easy to prove that the following
pairs of summands annihilate:
\begin{align}
  &+2g^{pa}T_{als}V^l_jV^s_rV^i_pV^r_k-2g^{pa}T_{als}V^l_jV^s_kV^i_rV^r_p=0,
  \\
  &+g^{ia}T_{pal}V^l_sV^s_rV^p_kV^r_j-g^{pa}T_{kal}V^l_sV^s_rV^i_pV^r_j=0,
  \\
  &-g^{ia}T_{ral}V^l_sV^s_pV^p_kV^r_j+g^{pa}T_{jal}V^l_sV^s_rV^r_kV^i_p=0.
  \\
  &+g^{pa}T_{kal}V^l_sV^s_jV^i_rV^r_p-g^{pa}T_{ral}V^l_sV^s_jV^i_pV^r_k=0,
  \\
  &-g^{pa}T_{jal}V^l_sV^s_kV^i_rV^r_p+g^{pa}T_{ral}V^l_sV^s_kV^i_pV^r_j=0.
\end{align}
\end{proof}

Summarizing, we have obtained that each second-order HHO defines a
multiparametric family of systems of conservation laws which is completely
characterized by Theorem~\ref{th:solut-comp-cond}. Since the velocity matrix of
the systems always has double eigenvalues in a generic situation
(Corollary~\ref{cor:eigenvalues}), according to Haantjes' Theorem
\cite{haantjes55:_x} the systems are diagonalizable if and only if there exist
a two-dimensional subspace of eigenvectors corresponding to every eigenvalue.

Unfortunately, we were not able to provide a generic proof on the
diagonalizability of the above systems. However, we can summarize the situation
in the following list.
\begin{description}
\item[$n=2$] All systems are linear: there is only one (non-degenerate)
  second-order HHO.
\item[$n=4$] All systems are linearizable by a projective reciprocal
  transformation, as the only nondegenerate representative of the second-order
  HHO in the classification is a constant coefficient $2$-form.
\item[$n=6$] Computational experiments show that for every equivalence class in
  the classification an explicit but random choice of the coefficients that
  determine the system leads to a diagonalizable velocity matrix.
\item[$n=8$] Again, computational experiments performed on randomly chosen
  members of the third family show that they are diagonalizable. More generic
  families need a more sophisticated computational approach in order to tame
  their complexity.
\end{description}
Given the results of the above computational experiments, we can
\emph{conjecture} that \emph{all} systems corresponding to second-order HHOs
are diagonalizable.

If the above conjecture is true, it turns out that all the above systems are
semi-Hamiltonians, as they are conservative and diagonalizable, by a theorem in
\cite{sevennec93:_geomet}. This already implies that such systems are
integrable, as they have infinitely many commuting flows (symmetries).

However, strictly speaking we cannot use the results of \cite{tsarev91:_hamil}
and conclude that the systems can be solved by the generalized hodograph
method, as it was developed for semi-Hamiltonian systems whose velocity matrix
has \emph{distinct} eigenvalues.

Despite the lack of general results on systems of our type, again computational
experiments show that we can solve randomly chosen systems corresponding to
second-order HHOs by the generalized hodograph method. So, it seems very
reasonable to make a stronger conjecture that systems corresponding to
second-order HHOs are \emph{integrable}.

\medskip

\textbf{Acknowledgements.} We thank R. Chiriv\`\i, E.V. Ferapontov,
P. Lorenzoni and M.V. Pavlov for useful discussions. This research has been
partially supported by the Department of Mathematics and Physics of the
Universit\`a del Salento, GNFM of the Istituto Nazionale di Alta Matematica
(INdAM), the research project Mathematical Methods in Non Linear Physics
(MMNLP) by the Commissione Scientifica Nazionale -- Gruppo 4 -- Fisica Teorica
of the Istituto Nazionale di Fisica Nucleare (INFN). PV aknowledges the support
of PRIN 2017 \textquotedblleft Multiscale phenomena in Continuum
Mechanics: singular limits, off-equilibrium and transitions\textquotedblright,
project number 2017YBKNCE.


\providecommand{\cprime}{\/{\mathsurround=0pt$'$}}
  \providecommand*{\SortNoop}[1]{}

\end{document}